\documentclass[10pt]{llncs} %
\usepackage[english]{babel}
\usepackage{a4}
\usepackage{latexsym}
\usepackage{amssymb}

\newenvironment{tab}{\begin{tabbing}
MMMMM\=aaa\=aaa\=aaa\=aaa\=aaa\=aaa\=aaa\=aaa\=aaa\=aaa\=aaa\=aaa\= \kill}{\end{tabbing}}

\makeatletter
\def\refmystepcounter#1{\stepcounter{#1}\protect\gdef 
\@currentlabel {\csname p@#1\endcsname \csname 
the#1\endcsname}}
\makeatother
\newcounter {tabnr}
\newenvironment{tabn}{\begin{tabbing}
\refmystepcounter{tabnr}
MMMMM\=aaa\=aaa\=aaa\=aaa\=aaa\=aaa\=aaa\=aaa\=aaa\=aaa\=aaa\=aaa\= \kill
(\arabic{tabnr})~}{\end{tabbing}}

\def\sem #1{\hbox{$[\![\, #1\, ]\!] $}}

\def\boks  {\mbox{$\Box$}}
\def\Nat   {\mbox{$\mathbb{N}$}}

\def\bar   {\mbox{$\,[ \! ]\,$}}

\def\all   {\forall\;}
\def\ex    {\exists\;}

\def\S #1/{\mbox {\textsl{#1}}}
\def\B #1/{\mbox {\textbf{#1}}}
\def\R #1/{\mbox {\textrm{#1}}}
\def\T #1/{\mbox {\texttt{#1}}}

\def\phi   {{\mbox{$\varphi$}}}

\def\Implies{\;\Rightarrow\;}

\def\EQ     {\mbox{\quad$\equiv$\quad}}

\def\Land   {\mbox{ $\;\land\;$ }}

\def\Lor    {\;\lor\;}

\def\TO     {\mbox{$\quad\to\quad$}}

\def\false  {\S false/}
\def\IS     {\mbox{$\quad =\quad $}}

\def\mbreak {\medbreak\noindent}

\begin{document}
\setcounter{tabnr}{-1}
\begin {center} {\Large Partial mutual exclusion for 
infinitely many processes}\\
  \mbox{}\\
  Wim H. Hesselink, \today\\
  Dept. of Computing Science,
  University of Groningen \\
  P.O.Box 407, 9700 AK Groningen, 
  The Netherlands\\
  Email: \verb!w.h.hesselink@rug.nl!\\
\end {center}

\begin{abstract}
  Partial mutual exclusion is the drinking philosophers problem for
  complete graphs. It is the problem that a process may enter a
  critical section \S CS/ of its code only when some finite set \S
  nbh/ of other processes are not in their critical sections. For each
  execution of \S CS/, the set \S nbh/ can be given by the
  environment. We present a starvation free solution of this problem
  in a setting with infinitely many processes, each with finite
  memory, that communicate by asynchronous messages.  The solution has
  the property of first-come first-served, in so far as this can be
  guaranteed by asynchronous messages.  For every execution of \S CS/
  and every process in \S nbh/, between three and six messages are
  needed.  The correctness of the solution is argued with invariants
  and temporal logic. It has been verified with the proof assistant
  PVS.
\end{abstract}

\mbreak Key words: Drinking philosophers; distributed algorithms;
starvation freedom; verification; fairness

\section {Introduction}

Partial mutual exclusion is a problem that goes back to Dijkstra's
dining philosophers \cite{Dij71} and the drinking philosophers of
Chandi and Misra \cite{ChM84}. It is the problem that a process may
enter a critical section of its code only when some specified finite
set \S nbh/ of other processes (neighbours) are not in their critical
sections. In the case of the dining philosophers, the philosophers
form a ring and \S nbh/ consists of the two neighbour philosophers.

The drinking philosophers form an arbitrary finite undirected graph,
say with $\S Nbh/.p$ as set of neighbours of philosopher $p$.  The set
$\S nbh/.p$ is then a modifiable subset of $\S Nbh/.p$, but the
algorithm has a message complexity proportional to the sizes of $\S
Nbh/.p$, and these can be considerably larger than $\S nbh/.p$.

Investigations into dining or drinking philosophers have always been
motivated by the desire for a clean abstraction of the problems of
resource allocation.

Inspired by the emergence of the internet, we generalize the setting
to potentially infinitely many processes and to arbitrary finite sets
\S nbh/ that can change over time. More precisely, when the
environment prompts a process, say $p$, to enter its critical section
\S CS/, it gives $p$ nondeterministically a finite set $\S nbh/.p$ of
other processes (to be called neighbours). After executing \S CS/,
process $p$ resets $\S nbh/.p $ to the empty set when it becomes idle
again.

\emph{Partial mutual exclusion} is the requirement that when two
processes, say $q$ and $r$, are both in each other's neighbourhoods,
they are not both in their critical sections:
\begin{tab}
\S PMX:/ \> $r\in\S nbh/.q \Land q\in\S nbh/.r 
\Land q\B\ in /\S CS/\Land r\B\ in /\S CS/ \Implies \false $ .
\end{tab}
We do not require that $r\in\S nbh/.q$ be equivalent to $q\in\S
nbh/.r$, because this would restrict the environment and the processes
considerably. When $q$ and $r$ are both in each other's
neighbourhoods, we speak of a \emph{conflict} between $q$ and $r$.
For comparison, mutual exclusion itself would be the requirement:
\begin{tab}
 \> $ q\B\ in /\S CS/\Land r\B\ in /\S CS/ \Implies q = r $ .
\end{tab}

We assume that the processes have private memory and communicate by
asynchronous messages. Messages are reliable: they are not lost or
duplicated.  They can pass each other, however. The processes receive
and answer messages even when they are idle.

We impose a \emph{first-come first-served} order (FCFS) in the
following way. Whenever a process starts the entry protocol, it
notifies all its neighbours.  If the notification of process $q$
reaches $r$ before process $r$ starts its entry protocol, and if the
entry of $r$ is in conflict with $q$, then process $q$ reaches \S CS/
before process $r$ does. The classsical definition of FCFS in
\cite{Lam74} only applies to shared memory systems and total mutual
exclusion. The above definition of FCFS seems to be the natural
translation for partial mutual exclusion with message passing.

As a process may have to wait a long time before it can enter the
critical section, it is important that the environment of any process
$p$ is allowed nondeterministically to abort the entry protocol of $p$
and move it back to the idle state.

There are two progress requirements: starvation freedom and maximal
concurrency. \emph{Starvation freedom} (\cite{ewd651}, also called
lockout freedom \cite{Lyn96}) means that, for every $p$, whenever
process $p$ needs to enter \S CS/, it will eventually do so if it is
not aborted.

Starvation freedom can only hold under the assumption of weak fairness
for all processes.  Weak fairness for process $p$ means that if, from
some time onward, process $p$ is continuously enabled to do a certain
step different from entering or aborting, it will do the step.  We
elaborate on the concept of weak fairness in Section \ref{s.liveness}.

The second progress requirement is \emph{maximal concurrency},
introduced in \cite{Rhe98}.  Maximal concurrency means that in any
case (i.e. without weak fairness) a process can make progress when it
has no conflicts with other processes.  More precisely, every process
$p$ that needs to enter \S CS/ and does not abort, will eventually
enter \S CS/, provided it satisfies weak fairness itself, all other
processes receive and answer messages from $p$, and no process comes
in an eternal conflict with $p$. The point here is that processes
without conflicts with $p$ need not be weakly fair; e.g., they are
allowed to remain in \S CS/ forever.

\subsection{Sketch of a solution}

The main reason for waiting is formed by \emph{conflicts}: $q\in\S
nbh/.r$ and $r\in\S nbh/.q$.  The conflict relation determines an
undirected graph.  The performance of any solution depends on the
sizes and the shapes of the connected components of this
graph. Conversely, however, the evolution of this graph depends on the
performance. The sooner a process reaches \S CS/, the sooner its set
\S nbh/ is made empty again.  This justifies the search for a simple
solution with as few messages as possible.

The first step of the solution was explicitly writing down what FCFS
would mean for partial mutual exclusion in a message passing system.
This led us to a solution with two layers: an outer protocol to
guarantee FCFS, and an inner protocol to guarantee partial mutual
exclusion. This is the same global design as used in the
shared-variable mutual exclusion algorithm of \cite{LyH91}.

The outer protocol requires 3 messages for every element of \S nbh/.
The inner protocol is asymmetric. We represent the processes by
natural numbers. If $q$ and $r$ are processes with $q<r$, we speak of
$q$ as the lower process and $r$ as the higher process.  The inner
protocol for process $p$ requires 3 messages for every higher process
in $\S nbh/.p$, and no messages for the lower processes in $\S
nbh/.p$.

\subsection{Restrictions} \label{CRMR}

Although we allow infinitely many processes, we restrict the number of
messages and the memory requirements of the processes in realistic
ways.

For the correctness of the algorithm, the time needed for message
transfer can be unbounded. For other issues, however, it is convenient
to postulate an upper bound $\Delta$ for the time needed to execute an
atomic command plus the time needed for the messages sent in this
command to be received.

We define the \emph{extended neighbourhood} of process $p$ to be the
union of $\S nbh/.p$ with its dual $ \S nbh/^0.p = \{q\mid p\in\S
nbh/.q\}$. The $k$th delayed extended neighbourhood of $p$ consists of
the processes that were in $p$'s extended neighbourhood not longer
than a time $k\Delta$ ago.

As there are many processes, we impose the communication restriction
(CR) that every process only sends messages to the members of some
delayed extended neighbourhood, and the memory restriction (MR) that
every process only remembers relationships with members of some
delayed extended neighbourhood. 

It is fairly easy to see that our solution has these properties, using
the first delayed extended neighbourhood for (CR) and the second one
for (MR).  We do not verify it formally, because this is cumbersome
and not illuminating.

\subsection{Overview and verification}

We briefly discuss related research in Section \ref{relatedresearch}.
In Section \ref{channels}, we describe the modelling of the
asynchronous messages.  Section \ref{algorithm} describes the
algorithm.  In Section \ref{safety}, we prove its safety properties
mutual exclusion and deadlock freedom.  In Section \ref{s.liveness},
we prove the liveness properties starvation freedom and maximal
concurrency. 

The proofs of the safety and liveness properties have been carried out
with the interactive proof assistant PVS \cite{OSR01}.  The
descriptions of proofs closely follow our PVS proof scripts, which can
be found on our web site \cite{HesUrlPartialMX}.  It is our intention
that the paper can be read independently, but the proofs require so
many case distinctions that manual verification is
problematic. Section \ref {PVS_sketch} provides a brief sketch of how
we use PVS. We conclude in Section \ref{conclusion}.

\subsection{Related research} \label{relatedresearch}

We solve the drinking philosophers' problem for an infinite complete
graph, as a clean idealization of the resource allocation problem.  We
do not intend to solve the general resource allocation problem itself,
as formulated in \cite{WeL93,AwS90,Rhe98}.

In the drinking philosophers' problems of \cite{ChM84,WeL93} (also
\cite[Section 20]{Lyn96}), the philosophers form a fixed finite
undirected graph. Our set $\S nbh/.p$ is then a subset of the constant
set $\S Nbh/.p$ of $p$'s neighbours in the graph.  This subset is
chosen by the environment each time that the process gets thirsty. The
message complexity of the solutions of these papers is proportional to
the size of $\S Nbh/.p$ and not to the possibly considerably smaller
size of $\S nbh/.p$.  These solutions are therefore problematic for
large complete graphs. To enforce starvation freedom, these solutions
assign directions to the edges of the graph such that the resulting
digraph is acyclic.

We are not aware of other solutions to the partial mutual exclusion
problem as formulated here.

The papers \cite{WeL93,Rhe98} offer a modular approach to the general
resource allocation problem.  This modular approach seems to
correspond to the outer protocol in our solution, but in either case,
the code is much more complicated than our outer protocol.  The
solution of \cite{AwS90} satisfies FCFS, but it is more complicated
and needs much more messages than our solution.

\subsection{Messages and channels} \label {channels}

As announced, we assume that the processes communicate by asynchronous
messages. Messages are guaranteed to arrive and be received, but
the delay is unknown, and messages are allowed to pass each other,
unlike in \cite{AwS90,Lyn96} where the messages in transit, say from
$q$ to $r$ are treated first-in-first-out.

Formally, we write $m.q.r$ to denote the number of messages $m$ sent
by $q$ to $r$. In particular, it is a natural number. Sending
corresponds to the incrementation $ \ m.q.r\T ++/ \ $.  Receiving
corresponds to decrementation $ \ m.q.r\T --/ \ $ and has the
precondition $m.q.r > 0$.  In principle, $m.q.r$ can be any natural
number, but in our algorithm, we take care to preserve the invariants
$m.q.r \leq 1$: there is always at most one message $m$ in transit
from $q$ to $r$.  Initially, no messages are in transit: $m.q.r=0$ for
every message type $m$.

In CSP \cite{Hoa85}, one would write $ \ m.q.r! \ $ for sending and $
\ m.q.r? \ $ for receiving, but messages in CSP are synchronous. In
Promela, the language of the model checker Spin \cite{Hol04}, one
could model the messages by channels with buffer size 1.

The view of the messages sent through channels must be taken with a
grain of salt, because it stretches the imagination to declare for
every process infinitely many channels.  For implementation purposes,
we prefer to regard a message $m$ sent by $q$ to $r$ as a tuple
$(m,q,r)$. Every process searches the tuple space continuously or
repeatedly for messages with itself as destination.

\section{The Algorithm} \label{algorithm}

The code for each process is decomposed as a parallel composition of 
three component processes: an environment, a forward stepping component, 
and a component that receives the messages:
\begin{tab}
  \> $ \B process/(p): $ \+\+\\
$ \B environ/(p)\;||\;\B forward/(p)\;||\;(||\;q: \B receive/(q, p)\:) $ .
\end{tab}
The component $\B environ/(p)$ implements the environment's decisions
for $p$: to start the entry protocol when it is idle, or to abort the
entry protocol if needful.  In component $\B forward/(p)$, process $p$
traverses the entry protocol towards \S CS/, followed by the exit
protocol back to the idle state.  For every $q$, component $\B
receive/(q, p)$ serves to receive and handle all messages from $q$ to
$p$.

Each component of process $p$ is an infinite loop, the body of which
is a nondeterministic choice between several guarded alternatives as
in Unity \cite{ChM88}.  The guard of an alternative is its enabling
condition. In $\B environ/(p)$ and $\B forward/(p)$, this is primarily
the value of the program counter \S pc/ of the process (called
\emph{line number}). In most cases of $\B receive/(q, p)$, the presence
of a message is the guard for its reception, and the first action is
the removal of the message. 

An important difference between $\B forward/(p)$ and $\B receive/(q,
p)$ as opposed to $\B environ/(p)$, is that the alternatives of $\B
forward/(p)$ and $\B receive/(q, p)$ are executed under weak fairness,
i.e., if in some execution one of its alternatives is for some point
onward continuously enabled this alternative will eventually be
executed. On the other hand, the environment is never forced to act.
We come back to this in Section \ref{s.liveness}.

One can argue that the critical section \S CS/ should be executed by
the environment. We reckon it to \B forward/ instead, because the
environment is allowed to do nothing, while \S CS/ needs to terminate.

\subsection{A layered solution}

If \S v/ is a private variable, outside the code, the value of \S v/
for process $q$ is denoted by $\S v/.q$.  In the algorithm, every
process has private variables \S nbh/, \S prio/, \S before/, \S
after/, \S wack/, \S away/, \S need/, \S prom/, which all hold finite
sets of processes, and which are initially empty. Every process has a
program counter $\S pc/:\Nat$, which is initially 11. The role of the 
ghost variable \S fork/ is explained in Section \ref{inner}.

We use five message types \T req/, \T gra/, \T notify/, \T withdraw/,
and \T ack/. As explained in Section \ref{channels}, $\T req/.q.r$
stands for the number of \T req/ messages from $q$ to $r$, and
analogously for the other message types.

The components of our solution are given in the Figures \ref{env},
\ref{fwd}, \ref{rcv}.  We can regard the alternatives as atomic
commands because actions on private variables give no interference,
the messages are asynchronous, and any delay in sending a message can
be regarded as a delay in message delivery.

\begin{figure}
\begin{tab}
  \> $ \B environ/(p): $ \+\+\\
  $ \bar\quad \S pc/ = 11 \TO \B choose /\S nbh/\B\ with /
  p\notin \S nhb/ \;;\; \S pc/ := 12 $ .\\
  $ \bar\quad \S pc/ = 12 \Land p\in \S AE/
  \TO \S nbh/ := \emptyset\;;\; \S pc/ := 11 $ .\\
  $ \bar\quad \S pc/ = 13 \Land p\in \S AE/ \TO $\\
  \>\> $ \B for /q\in\S nbh/ \B\ do /\T withdraw/.p.q\T ++/\B\ od/ $ ;\\
  \>\> $\S wack/ := \S nbh/ $ ; $\S nbh/ := \emptyset$ ;
  $\S prio/ := \emptyset$ ; $\S pc/ := 11 $ .\\
  $ \bar\quad \S pc/ = 14 \Land
  \S need/\cap\{q\mid p < q\} = \emptyset \Land p\in \S AE/ \TO$ \\
  \>\> $ \B for /q\in\S nbh/ \B\ with / p < q \B\ do / $\\
  \>\>\> $ \T gra/.p.q \T ++/\;;\; \S fork/.p.q\T --/\ \B\ od/$ ;\\
  \>\> $ \B for /q\in\S nbh/ \B\ do /\T withdraw/.p.q\T ++/\B\ od/ $ ;\\
  \>\> $\S wack/ := \S nbh/ $ ; $ \S need/ := \emptyset $ ; 
  $ \S nbh/ := \emptyset$ ; $ \S pc/ := 11 $ .
\end{tab}
\caption{The environment, to trigger and abort}\label{env}
\end{figure}

A process $p$ is called \S idle/ when $\S pc/.p=11$. When $p$ is idle,
the environment may decide to trigger the process by giving it a
finite set $\S nbh/.p$ with $p\notin \S nbh/.p$, and setting $\S pc/.p
:= 12$. See line 11 of \B environ/.  As $\S nbh/.p$ will only be
modified again when process $p$ becomes idle again, we postulate the
invariant
\begin{tab}
\S Iq0:/ \> $ r \in \S nbh/.q \Implies q\ne r \Land q\B\ in /\{12\dots\} $ .
\end{tab}
For a process $q$ and a line number $\ell$, we write $q\B\ at / \ell$
to express $\S pc/.q = \ell$. If $L$ is a set of line numbers, we
write $q\B\ in / L$ to express $\S pc/.q \in L$.  For all invariants,
we implicitly universally quantify over the free variables, usually
$q$ and $r$.

At line 12, the process starts the entry protocol in the component \B
forward/. When waiting at lines 12, 13, or 14 takes too long, the
environment may be allowed to abort \B forward/ and to go back to the
idle state. The fixed unspecified set \S AE/ (abort enabled) is
introduced to make it impossible that the aborting steps of the
environment are used to prove progress properties (by accident or
design).

As announced, the solution consists of two layers: an outer protocol
for FCFS, and an inner protocol for mutual exclusion. The partition in
layers is orthogonal to the partition in components. In $\B
forward/(p)$, the outer protocol is visible in line 12, in the guard
of line 13, and in the body of line 14.  The inner protocol is visible
in the body of line 13, in the guard of line 14, and in line 16 as a
whole.

\begin{figure}
\begin{tab}
  \> $ \B forward/(p): $ \+\+\\
  $ \bar\quad \S pc/ = 12 \Land \S wack/ = \emptyset \TO $ \\
  \>\> $ \B for /q\in\S nbh/ \B\ do /\T notify/.p.q\T ++/\B\ od/ $ ;\\
  \>\> $ \S prio/ := \S nbh/\cap(\S before/\setminus\S after/) $ ; 
  $ \S pc/ := 13 $ .\\

  $ \bar\quad \S pc/ = 13 \Land \S prio/ = \emptyset \TO $ \\
  \>\> $ \B for /q\in\S nbh/\B\ with /p < q
  \B\ do /\T req/.p.q\T ++/\B\ od/ $ ;\\
  \>\> $ \S need/:= \{q\in\S nbh/\mid p < q \Lor q\in \S away/\} $ ; 
  $ \S pc/ := 14 $ .\\
  
  $ \bar\quad \S pc/ = 14 \Land \S need/ = \emptyset \TO $ \\
  \>\> $ \B for /q\in\S nbh/ \B\ do /\T withdraw/.p.q\T ++/\B\ od/ $ ;\\
  \>\> $\S wack/ := \S nbh/ $ ; $ \S pc/ := 15   $ .\\
  $ \bar\quad \S pc/ = 15 \TO \S CS/ $ ; $ \S pc/ := 16 $ .\\

  $ \bar\quad \S pc/ = 16 \TO $ \\
  \>\> $ \B for /q\in\S nbh/ \B\ with / p < q \B\ do / 
  \T gra/.p.q \T ++/\;;\; \S fork/.p.q\T --/\ \B\ od/$ ;\\
  \>\> $ \S nbh/ := \emptyset $ ; $ \S pc/ := 11 $ .
\end{tab}
\caption{The stepping component}\label{fwd}
\end{figure}

\subsection{The outer protocol for FCFS}

In order to guarantee FCFS, every process $q$ maintains a set $\S
before/.q$ of the processes that have sent notifications to $q$
without withdrawing them.  Indeed, when a process $p$ has sent
notifications, it needs to withdraw them when it enters \S CS/, so
that $p$ does not force other processes needlessly to wait when it is
in its exit protocol or idle again.

Because the messages are asynchronous and not necessarily FIFO, the
message \T withdraw/ can arrive earlier than \T notify/.  We therefore
treat the messages \T notify/ and \T withdraw/ symmetrically. Arrival
of \T withdraw/ is registered in \S after/.  When both have arrived,
the combination is answered by a message \T ack/, to preclude
interference when \T notify/ or \T withdraw/ would be delayed. Each
process holds in \S wack/ the set of processes it is expecting
aknowledgements from. See the first four alternatives of $\B
receive/(q,p)$. Note that idle processes also accept messages.

According to the definition of FCFS, when process $p$ starts its entry
protocol in $\B forward/(p)$, it sends messages \T notify/ to its
neighbours, it forms a set $\S prio/.p$ as the intersection of $\S
nbh/.p$ with the difference between $\S before/.p$ and $\S after/.p$,
and then waits for the set $\S prio/.p$ to become empty. 

\begin{figure}[t]
\begin{tab}
  \> $ \B receive/(q, p): $ \+\+\\

  $ \bar\quad \T notify/.q.p > 0 \TO \T notify/.q.p \T --/ $ ;
  $ \B add /q\B\ to /\S before/ $ . \\

  $ \bar\quad \T withdraw/.q.p > 0 \TO \T withdraw/.q.p \T --/ $ ;\\ 
  \>\> $ \B remove /q \B\ from /\S prio/ $ ;
  $ \B add /q \B\ to /\S after/ $ .\\  

  $ \bar\quad q\in\S after/\cap\S before/ \TO $\\
  \>\> $ \B remove /q\B\ from /\S after/ 
  \B\ and /\S before/ $ ; $ \T ack/.p.q \T ++/ $ .\\

  $ \bar\quad \T ack/.q.p > 0 \TO \T ack/.q.p \T --/ $ ;
  $ \B remove /q\B\ from /\S wack/ $ . \\

  $ \bar\quad \T req/.q.p > 0 \TO \T req/.q.p \T --/ $ ;
  $ \B add /q \B\ to / \S prom/ $ .\\

  $ \bar\quad \T gra/.q.p > 0 \TO \T gra/.q.p\T --/ $ ;
  $ \S fork/.p.q \T ++/ $ ;\\
  \>\> $ \B remove /q \B\ from /\S away/\B\ and /\S need/ $ .\\

  $ \bar\quad q\in\S prom/\setminus\S away/\Land \neg\,
  (\S pc/ \geq 15 \Land q\in\S nbh/ ) \TO $\\
  \>\> $ \T gra/.p.q \T ++/\;;\;\S fork/.p.q \T --/$ ;\\
  \>\> $ \B add / q \B\ to / \S away/ $ ;
  $\B remove / q \B\ from / \S prom/ $ ;\\
  \>\> $ \B if /\S pc/ = 14 \Land q\in\S nbh/
  \B\ then / \B add / q \B\ to /\S need/ \B\ endif/ $ .

\end{tab}
\caption{The component of $p$ receiving messages from $q$}\label{rcv}
\end{figure}

When process $p$ enters \S CS/, it withdraws all its outstanding
notifications by sending \T withdraw/ messages to its neighbours,
because at that point it cannot be overtaken anymore.  It also sets
$\S wack/:=\S nbh/$.  When it arrives again at line 12, it waits for
\S wack/ to be empty. In this way, it verifies that all its messages
\T notify/ and \T withdraw/ have been acknowledged.

The third alternative of \B receive/ is called \T after/ because the
condition $q\in\S after/.p$ is its usual trigger.  This alternative can
be eliminated by including it conditionally in both the first and the
second alternative. We have not done so, because it would complicate
the code.  It may seem to be simpler to require separate
acknowledgements for \T notify/ and \T withdraw/. This, however, would
require more messages and more waiting conditions.

In early drafts of the algorithm, we had located the waiting for the
acks at line 16, as in Szymanski's algorithm \cite{Szy88}, but we have
rotated it to line 12 to avoid unnecessary waiting.

The outer protocol thus uses the messages \T notify/, \T withdraw/, \T
ack/, and the private variables \S prio/, \S before/, \S after/, and
\S wack/.

\subsection{The inner protocol} \label{inner}

The inner protocol serves to ensure partial mutual exclusion. It is
inspired by the drinking philosophers of \cite{ChM84} but we do not
insist on a symmetric solution. The idea is that the processes form a
modifiable directed complete graph, which need not remain acyclic.
For every edge, say between processes $q$ and $r$, the direction is
determined by a ``fork'' (or ``bottle'') that is held either by $q$ or
by $r$, or that is in transit between them. We use an integer variable
$\S fork/.q.r$ (private to $q$) to count the number of forks to $r$
that $q$ holds. 

In accordance with \cite{Dij71,ChM84}, we thus postulate the
invariants
\begin{tab}
\S Iq1:/ \> $ q \ne r \Implies\S fork/.q.r + \S fork/.r.q \leq 1 $ ,\\
\S Iq2:/ \> $ q \B\ in /\{15\dots\} \Land r\in\S nbh/.q \Implies 
\S fork/.q.r > 0$ .
\end{tab}
As \S CS/ is in line 15, condition \S PMX/ is clearly implied by \S
Iq0/, \S Iq1/, and \S Iq2/.  The values of $\S fork/.q.q$ are
irrelevant.

The basic idea of the inner protocol is that when process $p$ needs to
enter \S CS/, it sends request messages \T req/ to some neighbours $q$
for which it misses the fork, i.e., with $\S fork/.p.q = 0$. When it
has all forks needed, it can enter \S CS/. When process $p$ receives a
request for a fork that it does not need, it grants it by sending a \T
gra/ message.

At this point, we break the symmetry between processes, in two
ways. Recall that we represent the processes by natural numbers, and
that, if $q<r$, we say that process $q$ is \emph{lower} and that $r$
is \emph{higher}.  We decide to give priority to the lower process. It
follows that in the inner protocol, the lowest process waiting for
forks can make progress.

In view of the memory restriction (MR) of Section \ref{CRMR}, we
cannot use the infinite array $\S fork/.p$ as an actual private
variable of $p$. We therefore treat \S fork/ as a ghost variable,
which can be eliminated from the algorithm but only serves in its
description and its proof.  

We eliminate the variable \S fork/ from the algorithm by introducing
private variables to express how $\S fork/.p$ differs from a default
fork distribution to be introduced next. As the processes are
represented by natural numbers, there are two natural candidates for a
default fork distribution. Because the lower process has priority and
must therefore request the fork whenever it needs it, we locate the
fork by default at the higher process. The \emph{default fork
  distribution} therefore has $\S fork/.q.r = |r < q|$, where we use
$|b|=(b\,?\, 1: 0)$ for Boolean $b$. After every transaction this
state of affairs is restored as much as possible.

\begin{remark}
  It is possible to choose the alternative default fork distribution
  with $\S fork/.q.r = |q < r|$.  We come back to this in Section
  \ref{conclusion}. \boks
\end{remark}

\medbreak
The set of forks missing from the default distribution
for process $p$ is registered in the private variable $\S away/.p$. 
We thus postulate the invariant:
\begin{tab}
\S Iq3:/ \> $ q\in \S away/.r \EQ (q < r \Land \S fork/.r.q = 0) $ .
\end{tab}
Forks that are present despite the default fork distribution, are
present because the process asked for them. They are recorded in the
set difference $\S nbh/\setminus\S need/$.

In view of the default fork distribution, a process $p$ that needs
forks in line 13 sends requests \T req/ only to the higher neighbours
$q$. If process $q$ receives the request, it puts $p$ in its private
variable $\S prom/.q$ (promise).  When $p\in\S prom/.q$ and $q$ has
the fork, and is not currently using it, process $q$ sends the fork by
a message \T gra/, and updates its administration by adapting \S fork/
and \S away/. This last alternative of \B receive/ is called \T prom/.
It can be eliminated by including it conditionally in the alternatives
\T req/ and \T gra/, and in line 16.

When process $p$ receives a fork by a \T gra/ message, it accepts the
fork and updates the administration.  At line 16, process $p$ sends
the forks it has used back to its higher neighbours in accordance to
the default fork distribution.

To summarize, the inner protocol uses the messages \T req/ and \T
gra/, the private variables \S away/, \S need/, \S prom/, and the ghost
variables \S fork/.  Initially, the ghost variables \S fork/ satisfy
the default fork distribution $\S fork/.q.r = |r < q|$.

\subsection{Informal correctness arguments}

We postpone the full and formal proofs of correctness to Section
\ref{safety} for safety, and to Section \ref{s.liveness} for liveness.
Here, we only give some indications.

The proof that the inner protocol guarantees partial mutual exclusion
(\S PMX/) is a matter of careful fork administration. This is
relatively easy. The formal treatment is in Section \ref{mx_proof}.

Absence of deadlock is more complicated. There are three waiting
conditions at the lines 12, 13, and 14 of \B forward/, which each or
in combination potentially could lead to deadlock. The waiting
condition at line 12 is harmless, however, because it is just waiting
for messages to arrive. The waiting condition at line 13 does not lead
to deadlock, essentially because the process(es) waiting longest at
line 13 can proceed to line 14. The waiting condition at line 14 does
not lead to deadlock because the lowest process waiting at line 14 has
priority.  The formal treatment of deadlock is in Section
\ref{imm_dead}.

Starvation freedom is the most complicated property. The inner
protocol itself is not starvation free, because a lower process can
claim priority over a higher process. Note, however, that when it does
so, it will send the fork to the higher process in its exit protocol.
The outer protocol is starvation free because it satisfies FCFS. When
a process in the inner protocol is passed by another process, this
other process cannot pass again because it will be blocked by the FCFS
property of the outer protocol. As the number of processes in the
inner protocol is finite, it follows that the combination of the
protocols is starvation free.  The formal treatment is in Section
\ref{s.liveness}.

\subsection {Message complexity and waiting times} 
\label{summ_alg}

In the outer protocol, process $p$ exchanges 3 messages (\T notify/,
\T withdraw/, \T ack/) with every neighbour. In the inner protocol, it
exchanges 3 messages (\T req/, \T gra/, \T gra/) with every higher
neighbour. In total it exchanges between 3 and 6 messages with every
neighbour.

Component \B forward/ has 3 waiting conditions: emptiness of \S prio/
at line 13 to ensure FCFS, emptiness of \S need/ at line 14 to ensure
mutual exclusion, and emptiness of \S wack/ at line 12 to preclude
interference of delayed messages.

Recall from Section \ref{CRMR} that $\Delta$ serves as an upper bound
of the time needed for the execution of an alternative, plus the time
needed for reception of the messages sent.  Therefore, waiting for
emptiness of \S wack/ should not take more than $2\Delta$.

When the environment of $p$ wants to abort the entry protocol at line
14, it may need to wait for emptiness of the higher part of \S
need/. This waiting is also short because process $p$ has priority
over its higher neighbours. If $\Gamma$ is an upper bound for the
execution time of \S CS/, the higher part of \S need/ is empty after
$\Gamma + 2\Delta$.

The important waiting conditions are therefore emptiness of \S prio/
at line 13 and emptiness of \S need/ at line 14.  The first condition
is unavoidable and completely determined by FCFS. The waiting time for
emptiness of \S need/ at line 14 depends on the number of conflicting
processes that are concurrently in the inner protocol. The outer
protocol guarantees that conflicting processes do not enter the inner
protocol concurrently unless they are activated by the environment
within a period $\Delta$.  If the environment often activates several
conflicting processes within periods $\Delta$, our algorithm may have
performance problems. It seems likely, however, that other algorithms
would have the same problem.

\section{Verification of Safety} \label{safety}

In a distributed algorithm, at any moment, many processes are able to
do a step that modifies the global state of the system. In our view,
the only way to reason successfully about such a system is to analyse
the properties that cannot be falsified by any step of the system. These 
are the invariants. 

Formally, a predicate is called an \emph{invariant} of an algorithm if
it holds in all reachable states.  A predicate $J$ is called
\emph{inductive} if it holds initially and every step of the algorithm
from a state that satisfies $J$ results in a state that also satisfies
$J$.  Every inductive predicate is an invariant. Every predicate
implied by an invariant is an invariant.

When a predicate is inductive, this is often easily verified. In many
cases, the proof assistant PVS is able to do it without user
intervention. It always requires a big case distinction, because the
algorithm has 16 different alternatives in the Figures \ref{env},
\ref{fwd}, and \ref{rcv}.  

Most invariants, however, are not inductive.  Preservation of such a
predicate by some alternatives needs the validity of other invariants
in the precondition. We use PVS to pin down the problematic
alternatives, but human intelligence is needed to determine the useful
other invariants.

In proofs of invariants, we therefore use the phrase
``\emph{preservation of $J$ at $\ell_1\dots \ell_m$ follows from
  $J_1\dots J_n$}'' to express that every step of the algorithm with
precondition $J\land J_1\dots J_n$ has the postcondition $J$, and that
the additional predicates $J_1\dots J_n$ are only needed for the
alternatives $\ell_1\dots \ell_m$.  We indicate the alternatives of
Figure \ref{env} by \T env11/, \T env12/, \T env13/, \T env14/. The
alternatives of Figure \ref{fwd} are indicated by the line
numbers. The alternatives of Figure \ref{rcv} in which messages are
received, are indicated by the message names \T notify/, \T withdraw/,
\T ack/, \T req/, \T gra/. The alternatives 3 and 7 of \B receive/ are
indicated by \T after/ and \T prom/, respectively.

The \emph{follows from} relation makes the list of invariants into a
directed graph. In our enumerations of invariants, we traverse this
graph by breadth first search.

For all invariants postulated, the easy proof that they hold initially
is left to the reader. We use the term invariant in a premature
way. It will be justified at the end of the section.  

Section \ref{mx_proof} contains the proof that the algorithm satisfies
the invariant \S PMX/ of partial mutual exclusion.  Section
\ref{imm_dead} contains the proof of absence of deadlock.  This proof
uses invariants that are verified in Section \ref{inv_dead}.  

\subsection{The proof of mutual exclusion} \label{mx_proof}

In Section \ref{inner}, we saw that the mutual exclusion predicate \S
PMX/ is implied by \S Iq0/, \S Iq1/, and \S Iq2/.  This section
contains the proof that \S Iq0/, \S Iq1/, \S Iq2/, and \S Iq3/ of
Section \ref{inner} are invariants.  Firstly, it is easy to verify
that \S Iq0/ is inductive. Predicate \S Iq1/ is implied by the
observation that there is precisely one fork on every edge, as
expressed by the invariant:
\begin{tab}
\S Iq4:/ \> $ q \ne r \Implies 
\T fork/.q.r + \T fork/.r.q + \T gra/.q.r + \T gra/.r.q = 1 $ .
\end{tab}
Indeed, \S Iq4/ implies \S Iq1/ because \T gra/ holds
natural numbers.

Predicate \S Iq2/ is implied by the invariants:

\begin{tab}
\S Iq5:/ \> $ q\B\ in /\{14\dots\} \Land r\in\S nbh/.q
\Implies r\in\S need/.q \Lor \S fork/.q.r > 0 $ ,\\
\S Iq6:/ \> $ r\in\S need/.q \Implies q\B\ at / 14\Land r\in\S nbh/.q $ .
\end{tab}

The invariant \S Iq3/ is a matter of careful fork administration.
Preservation of \S Iq3/ at 16, \T gra/, and \T prom/ follows from \S
Iq2/, \S Iq4/, and the new invariants
\begin{tab}
  \S Iq7:/ \> $ \S fork/.q.r \geq 0$ ,\\
  \S Iq8:/ \> $ q\in\S prom/.r \Implies q < r $ .
\end{tab}
Note that \S Iq7/ is not superfluous, because, in the algorithm, we
unconditionally decrement $\S fork/.q.r$ in the alternatives 16 and \T
prom/. On the other hand, we treat the message variables $m.q.r$ as
natural numbers, because they are only decremented when positive.

Predicate \S Iq4/ is inductive.  Preservation of \S Iq5/ at 13 and \T
gra/ follows from \S Iq0/, \S Iq3/, and \S Iq7/.  Predicate \S Iq6/ is
inductive. Preservation of \S Iq7/ at 16, \T prom/, and \T env14/
follows from \S Iq2/, \S Iq3/, and \S Iq8/.  

Preservation of \S Iq8/ at \T req/ follows from the new invariant
\begin{tab}
  \S Iq9:/ \> $ \T req/.q.r > 0 \Implies q < r $ .
\end{tab}
Preservation of \S Iq9/ at 14 follows from \S Iq8/.

It now follows that the conjunction of the universal quantification of
the ``invariants'' introduced above is inductive.  Therefore, each of
them is itself invariant. In particular, the mutual exclusion
predicate \S PMX/ is invariant.  This concludes the proof that the
algorithm satisfies \S PMX/.

\subsection {Absence of deadlock} \label{imm_dead}

We define a state to be \emph{silent} when no process can do a step of
\B forward/ or \B receive/ of Figures \ref{fwd}, \ref{rcv}.  We define
a state to be \emph{in deadlock} when it is silent and some processes
are not idle (not at 11).  Note that the environment need not be
disabled.  Absence of deadlock is a safety requirement which will
follow from the liveness requirement of starvation freedom. It is
useful to prove absence of deadlock first, however, because the
ingredients of this proof are bound to enter again in the more
complicated proof of liveness.

For the proof of absence of deadlock, we need several additional
invariants.  Firstly, all processes are at the line numbers of the
program (otherwise there are no steps). This amounts to the inductive
invariant:
\begin{tab}
\S Jq0:/ \> $ q \B\ in /\{11\dots 16\}$ .
\end{tab}
Fork requests by the lower process are remembered as promises:
\begin{tab}
\S Jq1:/ \> $ q < r \Land r\in\S need/.q \Land \T req/.q.r = 
\T gra/.r.q = 0 \Implies q\in\S prom/.r $ .
\end{tab}
Predicate \S Jq1/ is invariant, because preservation at \T prom/ follows 
from \S Iq8/.

The following invariants are more difficult. We only claim them here,
and postpone the proofs to the next section.

Any process waiting for a fork, does not have it:
\begin{tab}
\S Jq2:/ \> $ r\in\S need/.q \Implies \S fork/.q.r = 0 $ .
\end{tab}
Because of the default fork distribution, a lower process has and gets
no fork unless it needs one:
\begin{tab}
\S Jq3:/ \> $ q < r \Land \S fork/.q.r +\T gra/.r.q > 0
\Implies q\B\ in /\{14\dots\} \Land r\in\S nbh/.q $ . 
\end{tab}
When process $q$ has sent \T withdraw/ to $r$, it expects an
acknowledgement from $r$.  It remembers this by putting $r$ in its set
\S wack/. This is expressed in the inductive invariant:
\begin{tab}
\S Kq0:/ \> $ \T withdraw/.q.r + \T ack/.r.q + |q\in\S after/.r|
= | r \in \S wack/.q | $ .
\end{tab}
When $q$ is in $\S after/.r$ and not in $\S before/.q$, a \T notify/
message is in transit from $q$ to $r$:
\begin{tab}
\S Kq1:/ \> $ q\in\S after/.r \Implies 
\T notify/.q.r > 0\Lor q\in\S before/.r $ .
\end{tab}

A binary relation $R$ is called \emph{well-founded}, if every nonempty
set $S$ has an $R$-minimal element, i.e., an element $q\in S$ such
that, for all $q'\in S$, we have $(q',q)\notin R$. 

We now claim the invariant that the relation $\S Prio/= \{(q, r)\mid
q\in\S prio/.r\}$ on the processes is always well-founded:
\begin{tab}
\S Lq0:/ \> $ \S well-founded/(\S Prio/) $ . 
\end{tab}
We also need: 
\begin{tab}
\S Lq1:/ \> $ q\in \S prio/.r \Land \T withdraw/.q.r = 0 \Implies
q \B\ in / \{13 \dots\}$ .
\end{tab}

Under assumption of these invariants, we can prove:

\begin{theorem} \label{thm-no-deadlock}
The system is deadlock free.
\end{theorem}

\begin{proof}
  Assume that the state is silent.  The idle processes are those at
  11. Because of \S Jq0/, the nonidle processes are at line 12 waiting
  for emptiness of \S wack/, at line 13 waiting for emptiness of \S
  prio/, or at line 14 waiting for emptiness of \S need/.  We have to
  prove that all processes are idle.

  Because the state is silent, there are no messages \T req/, \T gra/,
  \T notify/, \T withdraw/, and \T ack/.  By \S Iq4/, it follows that
  $\S fork/.r.q = 1 - \S fork/.q.r$ for all pairs $q\ne r$.  By \S
  Kq1/, the set $\S after/.r$ is a subset of $\S before/.r$ for all
  $r$.  As the alternative \S after/ is disabled for all $r$, it
  follows that $\S after/.r$ is empty for all processes $r$.  By \S
  Kq0/, it follows that $\S wack/.q$ is empty for all processes
  $q$. In particular, no process is disabled at line 12, and all
  processes are at the lines 11, 13, and 14.

  First assume that there are processes at 14.  Let $p$ be the lowest
  process at 14. As $p$ is disabled at line 14, the set $\S need/.p$
  is nonempty, say $q\in\S need/.p$.  By \S Jq2/, we have $\S
  fork/.p.q = 0$ and hence $\S fork/.q.p=1$. If $q < p$, then $q$ is
  at 11 or 13, contradicting \S Jq3/. Therefore $p < q $ by \S Iq6/
  and \S Iq0/. As there are no \T req/ or \T gra/ messages in transit,
  it follows that $p\in\S prom/.q$ by \S Jq1/. By \S Iq3/, we have
  $p\notin \S away/.q$.  As the alternative \S prom/ is disabled, it
  follows that $\S pc/.q \geq 15$, a contradiction. This proves that
  there are no processes at 14.

  We thus have that all disabled processes are at 13. Let $S$ be the
  set of processes at 13. If $S$ is empty, we are done. Therefore,
  assume that $S$ is nonempty. By \S Lq0/, the set $S$ has a \S
  Prio/-minimal element, say $p$. As $p$ is disabled at 13, the set
  $\S prio/.p$ is nonempty, say $q\in\S prio/.p$. As there are no \T
  withdraw/ messages, the invariant \S Lq1/ implies that $q$ is in
  $\{13\dots\}$. As all processes are at 11 or 13, this implies that
  $q\in S$ and $(q, p)\in\S Prio/$, contradicting the minimality of
  $p$.  Consequently, there are also no processes at 13. \boks
\end{proof}

\subsection {Invariants against deadlock} 
\label{inv_dead}

In this section we prove the invariants \S Jq2/, \S Jq3/, \S Kq0/, \S
Kq1/, \S Lq0/, \S Lq1/, postulated in the previous section. In the
course of this proof we need to postulate and prove several other
invariants. The reader may prefer not to verify these proofs, because
we obtained and verified them interactively with the proof assistant
PVS. What is relevant is to see the kind of assertions that can be
made and their logical relationships.

Preservation of \S Jq2/ at 13, 16 and \T prom/ follows from \S Iq2/, \S
Iq3/, \S Iq4/, \S Iq7/, \S Iq8/, and \S Jq3/. Preservation of \S Jq3/
follows at 16, \T prom/, \T env14/ from \S Iq4/, \S Iq6/, \S
Iq7/, and the invariant
\begin{tab}
  \S Jq4:/ \> $ q\in\S prom/.r\Implies r\in \S need/.q
  \Land \T req/.q.r = \T gra/.r.q = 0$ .
\end{tab}
Preservation of \S Jq4/ at 13, 16, \T req/, \T env14/ follows from \S
Iq6/, \S Iq8/, \S Jq5/ and the invariants
\begin{tab}
  \S Jq5:/ \> $ \T req/.q.r > 0 \Implies r\in \S need/.q \Land \T
  gra/.r.q = 0$ ,\\
\S Jq6:/ \> $ \T req/.q.r \leq 1 $ .
\end{tab}
Preservation of \S Jq5/ at 13, 16, \T prom/, \T env14/ follows
from \S Iq6/, \S Iq7/, \S Iq9/, \S Jq3/, and \S Jq4/.
Preservation of \S Jq6/ at 13 follows from \S Iq6/ and \S Jq5/. Note
that \S Jq4/ together with \S Iq8/ imply that the implication in \S
Jq1/ can be replaced by an equivalence.

Preservation of \S Kq0/ at 14, \T env13/, \T env14/ follows from the 
inductive invariant 
\begin{tab}
\S Kq2:/ \> $ q\B\ in /\{13, 14\} \Implies \S wack/.q = \emptyset $ .
\end{tab}

Preservation of \S Kq1/ at \T withdraw/ follows from the invariant
\begin{tab}
\S Kq3:/ \> $ \T withdraw/.q.r > 0 \Land \T notify/.q.r = 0
\Implies q\in\S before/.r $ .
\end{tab}

Preservation of \S Kq3/ at 14, \T env13/, \T env14/, and \T after/
follows from \S Kq0/ and the new invariant
\begin{tab}
\S Kq4:/ \> $ q\B\ in /\{13, 14\} \Land r\in\S nbh/.q\Land
\T notify/.q.r = 0 \Implies q\in\S before/.r $ .
\end{tab}
Preservation of \S Kq4/ at \T after/ follows from \S Kq0/ and \S Kq2/.

In order to prove \S Lq1/, we postulate the new invariant
\begin{tab}
\S Kq5:/ \> $ q\in\S before/.r \setminus\S after/.r 
\Land \T withdraw/.q.r = 0 $\\
\> $ \Implies q\B\ in /\{13, 14\} \Land r\in\S nbh/.q $ .
\end{tab}
Preservation of \S Kq5/ at \T notify/ follows from the new invariant:
\begin{tab}
\S Kq6:/ \> $ \T notify/.q.r > 0 
\Land q\notin \S after/.r \Land \T withdraw/.q.r = 0 $ \\
\> $ \Implies q\B\ in /\{13, 14\} \Land r\in\S nbh/.q $ .
\end{tab}
Preservation of \S Kq6/ at 14, \T env13/, \T env14/, \T after/ follows
from \S Kq0/ and the new invariant
\begin{tab}
\S Kq7:/ \> $ q\in\S before/.r \Implies \T notify/.q.r = 0 $ .
\end{tab}
Preservation of \S Kq7/ at 12 and \T notify/ follows from \S Kq0/, 
\S Kq5/, and the new invariant:
\begin{tab}
\S Kq8:/ \> $ \T notify/.q.r \leq 1 $ .
\end{tab}
Preservation of \S Kq8/ at 12 follows from \S Kq0/ and \S Kq6/.

We turn to preservation of \S Lq0/.  Recall that, for a relation $R$
and a set $S$, an element $s$ is called $R$-minimal in $S$ iff it
satisfies $s\in S$, and $(s',s)\notin R$ for all elements $s'\in
S$. Relation $R$ is called well-founded iff every nonempty set $S$ has
an $R$-minimal element. \S Lq0/ asserts that the relation \S Prio/,
which consists of the pairs $(q,r)$ with $q\in\S prio/.r$, is
well-founded.

It is easy to verify that, if $R$ is a well-founded relation, every
subrelation $R'\subseteq R$ is also well-founded. Therefore, \S Lq0/
is preserved by the modifications of \S prio/ in \T withdraw/ and \T
env13/, as these remove elements from \S Prio/. 

Preservation of \S Lq0/ at 12 follows from \S Iq0/, \S Kq0/, and \S
Lq1/. This is proved as follows. Assume that process $p$ executes line
12. Let us use \S Prio/ for relation \S Prio/ in the precondition, and
$\S Prio/^+$ for relation \S Prio/ in the postcondition of this step.
Assume that \S Lq0/ holds in the precondition.  Therefore, relation \S
Prio/ is well-founded.  It is easy to see that
\begin{tabn} \label{prio+}
\> $\S Prio/^+ \subseteq \S Prio/\cup\{(q, p)\mid q\in\S nbh/.p\}$ .
\end{tabn}

Let $S$ be a nonempty set of processes.  First, assume that
$S'=S\setminus \{p\}$ is nonempty. Therefore, $S'$ has a \S
Prio/-minimal element $r\in S'$. Because $p$ executes line 12, the set
$\S wack/.p$ is empty. Therefore \S Kq0/ implies $\T withdraw/.p.r =
0$, and \S Lq1/ implies $p\notin \S prio/.r$. Therefore $r$ is also \S
Prio/-minimal in $S$. It follows that $r$ is $\S Prio/^+$-minimal in
$S$ because of $r\ne p$ and formula (\ref{prio+}). 
It remains the case that $S\setminus \{p\}$ is empty. Then we have
$S=\{p\}$.  It follows that $p$ is a \S Prio/-minimal element of $S$
and, hence, a $\S Prio/^+$-minimal element of $S$ because of \S
Iq0/. 

In either case, the set $S$ contains a $\S Prio/^+$-mimimal
element. This proves that $\S Prio/^+$ is well-founded. Therefore \S
Lq0/ is preserved by the step of line 12.

Predicate \S Lq1/ is implied by the conjunction of \S Kq5/ and the
new invariant:
\begin{tab}
\S Lq2:/ \> $ \S prio/.q \subseteq \S before/.q\setminus \S after/.q $ .
\end{tab}
Predicate \S Lq2/ is inductive.  This concludes the proofs of the
invariants against deadlock.
\begin{remarks}
  Apart from \S Lq0/, all these invariants concern at most two
  processes. It is therefore possible to verify them by model
  checking. Such a system with two processes would not have too many
  states.  The proofs of \S Lq0/ and of absence of deadlock, however,
  do require theorem proving.

  The invariants of the lists \S Iq*/, \S Jq*/ relate to the inner
  protocol. The invariants of the lists \S Kq*/, \S Lq*/ are
  exclusively related to the outer protocol. The only invariant for
  both protocols is \S Iq0/.
\end{remarks}

\section{Liveness} \label{s.liveness}

The aim of this section is to show that the algorithm satisfies two
liveness properties: starvation freedom and maximal
concurrency. 

Starvation freedom means that every process that needs to enter \S CS/
and does not abort, will eventually enter \S CS/ and come back to the
idle state.  This can only be proved when all processes make progress
under weak fairness.

Maximal concurrency means that, even without global weak fairness
conditions, when process $p$ itself makes progress under weak fairness
and all processes answer messages from $p$, then whenever process $p$
needs to enter \S CS/ and does not abort, it will eventually enter \S
CS/ and come back to the idle state, unless it is, from some moment
onward, eternally in conflict with some other process. The latter case
is not deadlock, because the other process is not locked but only
fails to do steps as it is not subject to fairness conditions.

Weak fairness is introduced in Section \ref{intro-wf}. We formalize
executions as state sequences in Section \ref {execform}.  Section
\ref{inf_live} contains the formal definitions of weak fairness and
the statements of starvation freedom and maximal concurrency.  

The proofs of these two results are distributed over several
subsections.  In Section \ref{emptychan}, we introduce two more
invariants and show that all message channels are infinitely often
empty.  In Section \ref{loop_live}, we build the machinery to reduce
the proof obligations to progress at the lines 12, 13, and 14.  The
Sections \ref{progress12}, \ref{progress13}, \ref{progress14} treat
these lines. Section \ref{endofproofs} concludes the proofs of the two
theorems.

\subsection{Weak fairness} \label{intro-wf}

First, however, weak fairness needs an explanation.  Roughly speaking,
a system is called weakly fair if, whenever some process from some
point onward always can do a step, it will do the step.  Yet if a
process is idle, it must not be forced to be interested in \S
CS/. Similarly, if a process is waiting a long time in the entry
protocol, we do not want it to be forced to abort the entry
protocol. We therefore do not enforce the environment to do steps.  We
thus exclude the environment from the weak fairness conditions.

We impose the following weak-fairness conditions. For any process $p$,
if from some time onward $p$ can continuously do a \B forward/ step,
it will do the \B forward/ step. If, from some time onward, it can
continuously receive a message $m$ from some process $q$, it will
receive $m$ from $q$.

Formally, we do not argue about the fairness of systems, but
characterize the executions they can perform. Recall that an
\emph{execution} is an infinite sequence of states that starts in an
initial state and for which every pair of subsequent states satisfies
the step relation.  An execution is called \emph{weakly fair} iff, for
every process $p$, whenever $p$ can, from some state onward, always do
a \B forward/ step or receive a message $m$ from $q$, it will
eventually do the step or receive message $m$ from $q$. 

We also need that, when one of the alternatives \T after/ or \T prom/
of $\B receive/(q,p)$ is from some time onward continuously enabled,
this alternative is eventually taken. In the following, we therefore
treat these alternatives as if they correspond to messages $m = \T
after/$ or \T prom/ from $q$ to $p$.

\subsection{Formalization} \label{execform}

We formalize the setting in a set-theoretic version of (linear time)
temporal logic.  Let $X$ be the state space.  We identify the set
$X^\omega$ of the infinite sequences of states with the set of
functions $\Nat\to X$.  For a state sequence $\S xs/\in X^\omega$ and
$n\in\Nat$, we occasionally refer to $\S xs/(n)$ as the state a time
$n$. For a programming variable \T v/, we write $\S xs/(n).\T v/$ for
the value of \T v/ in state $\S xs/(n)$.

For a subset $U\subseteq X$, we define $\sem{U}\subseteq X^\omega $ as
the set of infinite sequences \S xs/ with $\S xs/(0)\in U$. For a
relation $A\subseteq X^2$, we define $\sem{A}_2\subseteq X^\omega$ as
the set of sequences \S xs/ with $(\S xs/(0),\S xs/(0))\in A$.

For $\S xs/\in X^\omega$ and $k\in \Nat$, we define the shifted
sequence $D(k, \S xs/)$ by $D(k, \S xs/)(n) = \S xs/(k+n)$.  For a
subset $P\subseteq X^\omega$ we define $\Box P$ (\emph{always} $P$)
and $\Diamond P$ (\emph{eventually} $P$) as the subsets of $X^\omega$
given by
\begin{tab}
\> $\S xs/\in \Box P \EQ (\all k\in\Nat: D(k, \S xs/)\in P) $  ,\\
\> $\S xs/\in \Diamond P \EQ (\ex k\in\Nat: D(k, \S xs/)\in P) $  .
\end{tab}

We now apply this to the algorithm.  We write $\S init/\subseteq X$
for the set of initial states and $\S step/\subseteq X^2$ for the step
relation on $X$. Following \cite{AbL91}, we use the convention that
relation \S step/ is reflexive (contains the identity relation).  An
\emph{execution} is an infinite sequence of states that starts in an
initial state and in which each subsequent pair of states is connected
by a step. The set of executions of the algorithm is therefore
\begin{tab}
\> $ \S Ex/ = \sem{\S init/}\cap \Box\sem{\S step/}_2$ .
\end{tab}

If $J$ is an invariant of the system, it holds in all states of every
execution. We therefore have $\S Ex/\subseteq \Box{J}$.

For our algorithm, the step relation $\S step/\subseteq X^2$ is the
union of the identity relation on $X$ (because \S step/ should be
reflexive) with the relations $\S step/(p)$ that consists of the state
pairs $(x,y)$ where $y$ is a state obtained when process $p$ does a
step starting in $x$.  The steps that process $p$ can do are
summarized in
\begin{tab}
\> $\S step/(p) \IS \S env/(p)\cup \S fwd/(p)\cup \bigcup_{q,m}
\S rec/(m,q,p)$ ,
\end{tab}
where $\S env/(p)$ consists of the steps of $\B environ/(p)$, $\S
fwd/(p)$ consists of the steps of $\B forward/(p)$, and $\S
rec/(m,q,p)$ consists of the steps where $p$ receives message $m$ from
$q$ in \B receive/. Note that we take the union here over all
processes $q$ and all seven alternatives $m$ of \B receive/ (including 
\T after/ and \T prom/).

We define $(q\B\ at /k)$ to be the subset of $X$ of the states in
which process $q$ is at line $k$. An execution in which process $q$ is
always eventually at \S NCS/, is therefore an element of
$\Box\Diamond\sem{q\B\ at /11}$. The aim is to prove that all
executions we need to consider are elements of this set.

\begin{remark} Note the difference between $\Box\Diamond\sem{U}$ and
  $\Diamond\Box\sem{U}$.  In general, $\Box\Diamond\sem{U}$ is a
  bigger set (a weaker condition) than $\Diamond\Box\sem{U}$. The
  first set contains all sequences that are infinitely often in $U$,
  the second set contains the sequences that are from some point
  onward eternally in $U$. \boks
\end{remark}

\subsection{Liveness under weak fairness} 
\label{inf_live}

For a relation $R\subseteq X^2$, we define the \emph{disabled} set
$D(R)=\{x\mid \all y: (x,y)\notin R\}$. Now \emph{weak fairness}
\cite{Lam94} for $R$ is defined as the set of state sequences in which
$R$ is disabled infinitely often or is taken infinitely often:
\begin{tab}
\> $ \S wf/(R) \IS \Box\Diamond \sem{D(R)} \cup \Box\Diamond \sem{R}_2$ .
\end{tab}
For a single process $p$, weak fairness for the steps of $\B
forward/(p)$ is the property $\S wf/(\S fwd/(p))$.

Our algorithm needs the property that every message, say $m$ in
transit from $q$ to $r$, is eventually received. The set $\S wf/(\S
rec/(m,q,r))$ contains precisely the state sequences that satisfy this
condition. This also applies to $m=\T after/$ and \T prom/.

For some purposes, we need the assumptions of weak fairness for the
steps of a single process $p$ and for all messages with $p$ as
destination or source. We thus define the set of \emph{executions weakly
fair for} $p$ by
\begin{tab}
  \> $ \S Wf/(p) = \S Ex/\cap\S wf/(\S fwd/(p))\cap\bigcap_{q, m}
(\S wf/(\S rec/(m,q,p))\cap \S wf/(\S rec/(m,p,q))) $ .
\end{tab}
The set of (globally) \emph{weakly fair executions} is defined by
messages, as captured in
\begin{tab}
  \> $ \S WF/ = \S Ex/\cap\bigcap_p\S wf/(\S fwd/(p))\cap\bigcap_{p, q, m}
\S wf/(\S rec/(m,q,p))$ .
\end{tab}

We can now formulate our two liveness results.  Starvation freedom
means that every process $p$ in every weakly fair execution is always
eventually back at \S NCS/ (i.e. at line 11). This is expressed by

\begin{theorem} \label{liveness} $ \S WF/\subseteq
  \Box\Diamond\sem{p\B\ at /11}$ \ for every process $p$.
\end{theorem}

Maximal concurrency means that every process $p$ that needs to enter
\S CS/ and does not abort, will eventually enter \S CS/, provided it
satisfies weak fairness itself, all other processes receive and answer
messages from $p$, and no process comes in an eternal conflict with
$p$.  Let us define $p\bowtie q$ to mean that $p$ and $q$ are in
conflict, i.e., $q\in\S nbh/.p\land p\in\S nbh/.q$.  Then maximal
concurrency is expressed by

\begin{theorem} \label{maxconcur} $ \S Wf/(p)\subseteq
  \Box\Diamond\sem{p\B\ at /11} \cup
  \bigcup_q\Diamond\Box\sem{p\bowtie q} $ \ for every process $p$.
\end{theorem}

The remander of this section is devoted to the proofs of these two
theorems. These proofs have a significant overlap. On the other hand,
the proof of Theorem \ref{liveness} has similarities with the proof of
absence of deadlock (Theorem \ref{thm-no-deadlock} in Section
\ref{imm_dead}).

\subsection{Empty channels} \label{emptychan}

At this point, we postulate two additional invariants:
\begin{tab}
\S Mq0:/ \> $ \T gra/.q.q = 0 $ ,\\
\S Mq1:/ \> $ r\in \S prio/.q \Implies q\B\ at /13 \Land r\in\S nbh/.q$ .
\end{tab}
Predicate \S Mq0/ is preserved at 16 and \T req/ because of \S Iq9/
and \S Iq8/. Predicate \S Mq1/ is inductive. 

We now claim that $m.q.r \leq 1$ always holds for all five message
types $m$ and all processes $q$ and $r$. For \T req/, \T notify/, \T
withdraw/, \T ack/, this follows from \S Jq6/, \S Kq8/, and \S Kq0/.
For \T gra/, it follows from \S Iq4/, \S Iq7/, and \S Mq0/.

As every state in every execution satisfies all invariants, and the
reception of message $m$ from $q$ by $r$ decrements $m.q.r$, it
follows that we have
\begin{tabn} \label{disMess}
\> $ \S Ex/\cap\S wf/(\S rec/(m,q,r)) \subseteq \Box\Diamond\sem{m.q.r = 0}$ .
\end{tabn}
In words, every message channel is infinitely often empty. 

One can do similar assertions about the alternatives \T after/ and \T
prom/, but this is not useful.

\subsection{Treating the loop} \label{loop_live}

We use the \emph{leads-to} relation between state predicates of
\cite{ChM88}.  A predicate $U$ is said to \emph{lead to} $V$ if it is
always the case that if $U$ holds, then eventually $V$ holds.  This is
formalized as follows.  For subsets $U$ and $V$ of $X$, the set of
state sequences in which $U$ \emph{leads to} $V$ is defined by
\begin{tab}
\> $ \S LT/(U,V) = \Box(\neg \sem{U}\cup \Diamond\sem{V}) $ .
\end{tab}

Our specific algorithm is a simple loop with the property that, in any
execution, if some process $p$ does not get stuck at a line $k$,
it will eventually proceed to line $k+1$ or to 11. 

We need to prove that a process reaches 11 from a combination of
lines. We therefore define
\begin{tab}
\> $ \S toIdle/(k, p) = \S LT/(p \B\ in / \{k\dots\}, p \B\ at /11) $ .
\end{tab}
The relevance of this concept follows from the inclusion:
\begin{tabn} \label{basis}
\> $ \S Ex/\cap \S toIdle/(12, p) \subseteq \Box\Diamond\sem{p\B\ at /11}$ .
\end{tabn}
This holds because, in any execution that belongs to $\S toIdle/(12,
p)$, if at some time $p$ is not at 11, then $p$ is in $\{12\dots\}$ by
\S Jq0/, and therefore $p$ will return to 11.

On the other hand, in any execution, process $p$ is never in
$\{17\dots\}$ by \S Jq0/. We therefore have
\begin{tabn} \label{basis17}
\> $ \S Ex/\subseteq \S toIdle/(17, p) $ .
\end{tabn}
The aim is thus to decrement the first argument of \S toIdle/. This is
done with the relation
\begin{tabn} \label{reduce}
\> $ \S Ex/ \cap \S toIdle/(k+1, p) \subseteq \S toIdle/(k, p)
\cup\Diamond\Box\sem{p\B\ at /k}$ .
\end{tabn}
This formula just means that, in any execution, a process $p$ that is
ever at line $k$, but from that time onward never at line 11 or in
$\{k+1\dots\}$, needs to remain at $k$.

As a process is never disabled at the lines 15 or 16, weak fairness
implies that it never stays there, i.e., we have
\begin{tab}
  \> $ \S Ex/\cap \S wf/(\S fwd/(p)) \cap \Diamond\Box\sem{p\B\ at /k}
  = \emptyset $ \ for  $k=15 $, 16.
\end{tab}
Therefore, the formulas (\ref{basis17}) and (\ref{reduce}) imply that
\begin{tabn} \label{from15}
\> $ \S Ex/\cap \S wf/(\S fwd/(p)) \subseteq \S toIdle/(15, p)$ .
\end{tabn}
It remains to eliminate the executions in $\Diamond\Box\sem{p\B\ at /k}$
for $k=12$, 13, and 14. This is done in the next three sections. 

\subsection{Progress at line 12} \label{progress12}

Consider an execution $\S xs/ \in \Diamond\Box\sem{p\B\ at /12}$.
From some time $n_0$ onward, process $p$ is and remains at line
12. Therefore, weak fairness of $\S fwd/(p)$ implies that it is
infinitely often disabled. It is disabled iff the set $\S wack/.p$ is
nonempty. While $p$ is at line 12, the finite set $\S wack/.p$ can
only become smaller. It therefore is eventually constant and
nonempty. So, there is a process $q$ such that, from time $n_0$
onward, $q\in \S wack/.p$ always holds.

Assume that $\T ack/.q.p > 0$ holds at some time $n_1\geq n_0$. Then, by
formula (\ref{disMess}), there is a time $n\geq n_1$ such that $\T
ack/.q.p=0$; therefore the \T ack/ message is received and $q\in\T
wack/.p$ is falsified. This proves that $\T ack/.q.p = 0$ holds at any
time $n\geq n_0$.

By formula (\ref{disMess}), there is a time $n_1\geq n_0$ such that
$\T withdraw/.p.q = 0$. By \S Kq0/, we then have $p\in\S
after/.q$. This can only be falsified by the alternative \T after/,
which sends a message $ \T ack/.q.p$. We therefore have that $p\in\S
after/.q$ holds at all time $n\geq n_1$. By formula (\ref{disMess}),
there is a time $n_2\geq n_1$ such that $\T notify/.p.q = 0$ and hence
$p\in\S before/.q$ by \S Kq1/.  This can only be falsified by the
alternative \T after/ which sends a message $ \T ack/.q.p$. We
therefore have that $p\in\S before/.q$ holds at all time $n\geq n_2$.
Therefore, the alternative \T after/ is eternally enabled from time
$n_2$ onward. By weak fairness it will be taken, thus falsifying
$p\in\S before/.q$. This is a contradiction.

We have derived this contradiction using weak fairness of $\S
fwd/(p)$, and of $\S rec/(m, p, q)$ and $\S rec/(m, q, p)$ for all $q$
and some $m$.  We therefore have proved that
\begin{tabn} \label{at12evtNot}
\> $ \S Wf/(p) \cap \Diamond\Box\sem{p\B\ at /12} = \emptyset$ .
\end{tabn}

\subsection{Progress at line 13} \label{progress13}

We now want to exclude the possibility that in some execution some
process is eventually always at line 13.  For an execution \S xs/ and
a time $n\in\Nat$, let $S(n, \S xs/)$ be the set of processes that, in
\S xs/, from time $n$ onward, is always at 13.  The invariant \S Lq0/
in the state $\S xs/(n)$ implies:
\begin{tabn} \label{at13notPrio}
\> $ \S xs/\in\S Ex/\Land S(n, \S xs/)\ne\emptyset $\\
\> $ \Implies \ex p\in S(n, \S xs/): 
   S(n, \S xs/)\cap\S xs/(n).\S prio/.p=\emptyset $ .
\end{tabn}

Let $p\in S(n, \S xs/)$.  From time $n$ onward, process $p$ is and
remains at line 13.  By weak fairness of $\S fwd/(p)$, process $p$ it
is infinitely often disabled.  As it is at line 13, process $p$ is
disabled iff the set $\S prio/.p$ is nonempty. While $p$ is at line
13, the finite set $\S prio/.p$ can only become smaller. It therefore
is eventually constant and nonempty. So, there is a process $q$ such
that from time $n$ onward, $q\in \S prio/.p$ always holds.  It follows
that, from time $n$ onward, process $p$ does not receive \T withdraw/
from $q$. By weak fairness of $\S rec/(\T withdraw/, q, p)$, it
follows that $\T withdraw/.q.p=0$ holds from time $n$ onward.  We thus
have proved:

\begin{tabn} \label{at13prio} 
  \> $ \S xs/\in\S Ex/\Land p\in S(n, \S xs/) 
  \Land \S xs/\in\S Wf/(p) $ \\
  \> $ \Implies \ex q: \all i: q\in\S xs/(n+i).\S prio/.p \Land 
  \S xs/(n+i).\T withdraw/.q.p=0 $ .
\end{tabn}
At this point, we note that the invariants \S Kq5/, \S Lq2/, and \S
Mq1/ together imply:
\begin{tabn} \label{prioGives}
\> $ q\in \S prio/.p \Land \T withdraw/.q.p=0 \Implies
q\B\ in /\{13, 14\}\Land p\bowtie q $ .
\end{tabn}
Therefore, formula (\ref{at13prio}) implies
\begin{tabn} \label{at13conflict}
\> $ \S Wf/(p) \cap\Diamond\Box\sem{p\B\ at /13}
\subseteq \bigcup _q \Diamond\Box \sem{ p\bowtie q} $ .
\end{tabn}

For the sake of starvation freedom, we combine (\ref{at13notPrio}) 
and (\ref{at13prio}) to
\begin{tabn} \label{at13_14}
\> $ \S WF/ \cap\Diamond\Box\sem{r\B\ at /13}
\subseteq \bigcup _q \Diamond\Box \sem{ q\B\ at /14} $ .
\end{tabn}
This formula is proved as follows. Let \S xs/ be in the lefthand set.
Then there is a number $n$ such that $S(n, \S xs/)$ is
nonempty. Formula (\ref{at13notPrio}) gives some process $p\in S(n, \S
xs/)$ with $\S prio/.p$ disjoint from $S(n, \S xs/)$.  As $\S
WF/\subseteq \S Wf/(p)$, the formulas (\ref{at13prio}) and
(\ref{prioGives}) yield a process $q$ that is and remains in $\S
prio/.p$ and that is and remains in $\{13, 14\}$. As $q\notin S(n, \S
xs/)$, at follows that $q$ is eventually always at 14.

\subsection{Progress at line 14} \label{progress14}

Let \S xs/ be an execution in $\Diamond\Box\sem{p\B\ at /14}$.
Process $p$ waits at line 14 for emptiness of \S need/. This condition
belongs to the inner protocol.  The inner protocol in isolation,
however, is not starvation free because it would allow a lower process
repeatedly to claim priority over $p$ by sending \T req/s.  We need
the FCFS property of the outer protocol to preclude this.
Technically, the problem is that $\S need/.p$ can grow at line 14 in
the alternative \T prom/.

Partly, in order to prove that eventually the truth value of $q\in\S
need/.p$ is constant, we construct a numeric state function $\S
vf/(q,p) \geq 0$ that, for $q\in\S nbh/.p$, eventually never increases
and therefore stabilizes to a constant value, and that it is only
constant when the truth value of $q\in\S need/.p$ is also constant. We
construct \S vf/ as the weighted sum of three bit-valued state
functions:
\begin{tab}
\> $ \S vf/(q, p) =  \S vf0/(q,p) + 2\cdot \S vf1/(q,p) 
+ 4\cdot \S vf2/(q,p)  $ \ where \\
\>\> $ \S vf0/(q, p) = |\, q\in\S need/.p\, | $ ,\\
\>\> $ \S vf1/(q, p) = |\, \S fork/.q.p+\T gra/.p.q = 0 \Land q < p\,| $ ,\\
\>\> $ \S vf2/(q, p) = |\, q\B\ in /\{13\dots\} \Land p\in\S nbh/.q
\Land p\notin\S prio/.q\,| $ .
\end{tab}  
Indeed, when process $p$ is and remains at line 14, its neighbourhood
$\S nbh/.p$ is constant. For any $q\in\S nbh/.p$, we have eventually
$\T notify/.p.q = 0$ by formula (\ref{disMess}). This remains valid
because $p$ at 14 does not send \T notify/. By \S Kq4/, we therefore
have eventually always $p\in\S before/.q$.

We claim that, while $p$ is at line 14 and $p\in\S before/.q$ holds,
$\S vf/(q,p)$ never increases.  This is proved as follows. At line 12,
$\S vf2/(q,p)$ does not increase because of \S Kq0/ and \S Kq2/. The
same holds for \T withdraw/. At line 16, $\S vf1/(q,p)$ can be
incremented, but this is compensated by decrementation of $\S
vf2/(q,p)$ because of \S Mq1/. The difficult alternative is \T prom/,
because it can increment $\S vf0/(q,p)$ by adding $q$ to $\S
need/.p$. In that case, \S Iq8/ implies $q<p$. Therefore, $\S
vf1/(q,p)$ is decremented because of \S Iq3/, \S Iq4/, and \S Iq7/.
This proves the claim.

It follows that, if process $p$ is and remains at line 14, eventually
$\S vf/(q, p)$ becomes constant.  When $\S vf/(q, p)$ is constant,
$q\in\S need/.p$ is also constant because $q\in\S need/.p$ holds if
and only if $\S vf/(q, p)$ is odd.  This proves that, eventually, the
truth value of $q\in\S need/.p$ becomes constant.

As $\S need/.p$ is always a subset of the finite set $\S nbh/.p$,
which is constant while $p$ is at line 14, we can now conclude that
eventually $\S need/.p$ is constant. If this constant would be the
empty set, process $p$ would be eventually always enabled, and by weak
fairness, process $p$ would leave line 14. Therefore, there is some
process $q$ eventually always in $\S need/.p$. We have $q\ne p$
because of \S Iq0/ and \S Iq6/. This proves
\begin{tabn} \label{evtNeed} 
  \> $ \S Wf/(p) \cap \Diamond\Box\sem{p\B\ at /14} 
  \subseteq \bigcup_{q\ne p} \Diamond\Box \sem{q\in\S need/.p} $ .
\end{tabn}

We now distinguish the cases $q<p$ and $p < q$. For the first case, we
claim:
\begin{tabn} \label{14low}
  \> $ q < p \Implies \S Wf/(p) \cap \Diamond\Box \sem{q\in \S need/.p} 
  \subseteq \Diamond\Box \sem{q\B\ in /\{14\dots\}\land p\in\S nbh/.q} $ .
\end{tabn}
This is proved as follows. Firstly, $q\in\S need/.p$ implies $\S
fork/.p.q=0$ by \S Jq2/.  By (\ref{disMess}) we have infinitely often
$\T gra/.q.p=0$. Therefore, \S Jq3/ and \S Iq4/ imply that the
conjunction $q\B\ in /\{14\dots\}\land p\in\S nbh/.q$ holds infinitely
often.

When process $q$ leaves $\{14\dots\}$ by executing line 16 or \T
env14/, it decrements $\S vf2/(q,p)$ because of \S Mq1/; it therefore
also decrements $\S vf/(q,p)$ (even if $\S vf1/(q,p)$ increases).  As
$\S vf/(q,p)$ is eventually constant, it follows that eventually $q$
remains in $\{14\dots\}$ and therefore $p$ remains in $\S
nbh/.q$. This proves (\ref{14low}).

For the second case, we claim: 
\begin{tabn} \label{14high} 
\> $ p < q \Implies \S Wf/(p) \cap
\Diamond\Box \sem{q\in \S need/.p} \subseteq \Diamond\Box \sem{q\B\ in
/\{15\dots\}\land p\in\S nbh/.q } $ .
\end{tabn}
This is proved as follows.  The fact that $q$ remains in $\S need/.p$
together with \S Iq6/ and formula (\ref{disMess}) are used to prove
that eventually always $\T req/.p.q = 0$, and $ \T gra/.q.p = 0$, and
$ \T gra/.p.q = 0$.  This implies eventually always $p\in\S prom/.q$
by \S Jq1/, and $p\notin\S away/.q$ by \S Jq2/, \S Iq3/, and \S
Iq4/. Yet, the alternative \T prom/ is not taken anymore.  Therefore,
weak fairness implies that $ q\B\ in /\{15\dots\}\land p\in\S nbh/.q$
holds infinitely often. Again using that $\S vf/(q,p)$ is eventually
constant, we get that process $q$ eventually remains in $\{15\dots\}$
and that $p$ remains in $\S nbh/.q$. This proves (\ref{14high}).

As $q\in\S need/.p$ implies $q\in \S nbh/.p$ by \S Iq6/, the formulas 
(\ref{evtNeed}), (\ref{14low}), (\ref{14high}) combine to yield 
\begin{tabn} \label{at14conflict}
\> $ \S Wf/(p) \cap\Diamond\Box\sem{p\B\ at /14}
\subseteq \bigcup _q \Diamond\Box \sem{ p\bowtie q} $ .
\end{tabn}

With regard to starvation freedom, we use the formulas
(\ref{evtNeed}), (\ref{14low}), (\ref{14high}) to prove
\begin{tabn} \label{not14}
  \> $ \S WF/ \cap \Diamond\Box\sem{r\B\ at /14} = \emptyset$ .
\end{tabn}
This is done by contradiction. Assume that \S xs/ is an element of the
lefthand expression. Let $p$ be the lowest process with $\S xs/ \in
\Diamond\Box\sem{p\B\ at /14}$. Formula (\ref{evtNeed}) gives us a
process $q\ne p$ with $ \S xs/\in\Diamond\Box\sem{ q\in\S need/.p}$.
By (\ref{from15}), we have
\begin{tabn} \label{obv15}
\> $\S xs/\in\S toIdle/(15, q)$ . 
\end{tabn}
If $q < p$, formula (\ref{14low}) gives $\S xs/\in \Diamond\Box
\sem{q\B\ in /\{14\dots\} } $. Together with (\ref{obv15}), this
implies $\S xs/\in \Diamond\Box\sem{q\B\ at /14}$, contradicting the
minimality of $p$. If $p < q$, formula (\ref{14high}) gives $\S xs/\in
\Diamond\Box \sem{q\B\ in /\{15\dots\} } $, which contradicts
(\ref{obv15}).  This concludes the proof of (\ref{not14}).

\subsection{End of proofs} \label{endofproofs}

Theorem \ref{liveness} follows from the formulas (\ref{basis})
and (\ref{from15}) by repeated application of (\ref{reduce})
with (\ref{at12evtNot}), (\ref{at13_14}), (\ref{not14}). 

Similarly, Theorem \ref{maxconcur} follows from the formulas
(\ref{basis}) and (\ref{from15}) by repeated application of
(\ref{reduce}) with (\ref{at12evtNot}), (\ref{at13conflict}), 
 (\ref{at14conflict}).

\section{Sketch of the PVS verification} \label{PVS_sketch}

The reader who is familiar with PVS can obtain the dump of the proof
script from our website \cite{HesUrlPartialMX}. The proof consists of
183 lemmas that can be verified within three minutes on an ordinary
laptop.

For the reader who is not familiar with proof assistants, we give some
indications here how we used the proof assistant PVS.

The first thing to do is to declare the state space of the algorithm. 
This is done by means of the declarations:
\begin{verbatim}
  Process: TYPE FROM nat
\end{verbatim}
\begin{verbatim}
  state: TYPE = [#
    fork: [Process -> [Process -> int]], 
    req, gra, notify, withdraw, ack: [Process -> [Process -> nat]],
    pc: [Process -> nat],
    nbh, need, prom, away, wack, before, prio, after: 
      [Process -> finite_set[Process]]
  #]
\end{verbatim}
The first line says that \T Process/ is an unspecified subset of
$\Nat$.  The second line declares \T state/ as a record with the
program variables as fields.

In order to inform PVS of the types of the free variables that we are
going to use, we declare
\begin{verbatim}
  p, q, r: VAR Process
  x, y: VAR state
\end{verbatim}

The next point is to define the step relation of the transition system
according to Figures \ref{env}, \ref{fwd}, \ref{rcv}. We construct it
as the union of separate relations for each of the 16 alternatives.
As an example, the state modification at line 12 of \B forward/ is
given by
\begin{verbatim}
  next12(p, x): state =
    x WITH [
      `notify(p) := LAMBDA q: x`notify(p)(q) + b2n(x`nbh(p)(q)),
      `prio(p) := {q | x`nbh(p)(q) AND x`before(p)(q) 
                       AND NOT x`after(p)(q)},
      `pc(p) := 13
    ]
\end{verbatim}
Here \T x`nbh(p)(q)/ means $q\in\S nbh/.p$ in state \T x/, and \T b2n/
is the function that converts a Boolean to a bit. The corresponding
step relation is:
\begin{verbatim}
  step12(p, x, y): bool =
    x`pc(p) = 12 AND y = next12(p, x) AND empty?(x`wack(p))
\end{verbatim}

When the step relation has been constructed, we turn to the
construction and the verification of the invariants. This is a major
effort. One of the simpler cases is \S Iq8/.
\begin{verbatim}
  iq8(q, r, x): bool =
    x`prom(r)(q) IMPLIES q < r
\end{verbatim}
\begin{verbatim}
  iq8_rest: LEMMA 
    iq8(q, r, x) AND step(p, x, y) 
    IMPLIES iq8(q, r, y) OR stepReq(p, x, y)
\end{verbatim}
\begin{verbatim}
  iq8_Req: LEMMA 
    iq8(q, r, x) AND stepReq(p, x, y) AND iq9(q, r, x)
    IMPLIES iq8(q, r, y)
\end{verbatim}
\begin{verbatim}
  iq8_step: LEMMA 
    iq8(q, r, x) AND step(p, x, y) AND iq9(q, r, x)
    IMPLIES iq8(q, r, y)
\end{verbatim}
This shows that preservation of \S Iq8/ only needs \S Iq9/ at the
alternative \T req/. When all invariants separately have been
constructed, we form the conjunction \T globinv/ of the universal
quantifications of them and prove that \T globinv/ is inductive:
\begin{verbatim}
  globinv_step: LEMMA
    globinv(x) AND step(p, x, y) 
    IMPLIES globinv(y)
\end{verbatim}
\begin{verbatim}
  globinv_start: LEMMA
    start(x) IMPLIES globinv(x)
\end{verbatim}

We formalize state sequences, executions, and weak fairness just as
described in Section \ref{s.liveness}. For instance, function \S wf/
of Section \ref{inf_live} is given by
\begin{verbatim}
  weakly_fair(rel)(xs): bool =  
    box(diamond(sem1(disabled(rel))))(xs) 
    OR box(diamond(sem2(rel)))(xs)
\end{verbatim}
Finally, Theorem \ref{liveness} is given by
\begin{verbatim}
  liveness: THEOREM  % starvation freedom
    weakly_fair_all(xs) AND execution(xs) 
    IMPLIES box(diamond(sem1(at(11, p))))(xs)
\end{verbatim}
 
\section{Conclusions} \label{conclusion}

The acyclicity invariant introduced by Chandy and Misra for their
drinking philosophers \cite{ChM84}, was not necessary for liveness,
and it was problematic for the message complexity and memory
requirements.  In our solution, we abandon the acyclicity invariant.
We also break the symmetry, in two ways. Firstly by giving higher
priority to the lower processes. Secondly by using an asymmetric
default fork distribution, in which the higher processes are better
off in the sense that they need not request forks.

Originally, we had a weaker outer protocol that only added starvation
freedom to the inner protocol, and the inner protocol used the low
default fork distribution with $\S fork/.q.r = |q < r|$.  When we had
completed its proof, we disliked the order in which processes could
enter \S CS/.  We therefore investigated the high default, and
postulated the FCFS property.  It came as a surprise to us that this
resulted in a somewhat simpler algorithm.

The design of the algorithm was only possible, because we could use
the proof assistant PVS \cite{OSR01} to verify the invariants and the
progress requirements.  We could fruitfully reuse parts of the
PVS-proof of the earlier algorithms in the proof of the final
algorithm.

We hope that the algorithm or variations of it can be used in the
design of practical resource allocation algorithms. Indeed, when the
processes compete for several resources, Chandy and Misra \cite{ChM84}
suggest to use coloured bottles (forks) with a colour for every
resource.  This can also be done in our algorithm. There should,
however, be some relationship between the resources in the application
and the neighbourhoods in our model, and this relationship should be
exploited for better performance in the presence of many processes.
For applications on the internet, one would need to extend the
algorithm with fault tolerance.  These extensions are matters for
future research.

\bibliographystyle{plain} %{alpha} %
\bibliography{refs}

\end{document}